\documentclass[letterpaper, 10pt, conference]{ieeeconf}
\IEEEoverridecommandlockouts
\usepackage{cite}
\usepackage{amsmath,amssymb,amsfonts}
\usepackage{algorithm}
\usepackage{graphicx}
\usepackage{subfigure}

\usepackage{tikz}
\usepackage{textcomp}
\usepackage{xcolor}
\usepackage{algpseudocode}

\usepackage{lscape}
\usepackage{siunitx}

\newtheorem{proposition}{Proposition} [section]

\newtheorem{definition}{Definition}[section]

\begin{document}

\title{\bfseries Extended Kalman Filtering for Recursive\\ Online Discrete-Time Inverse Optimal Control}%
\author{Tian Zhao and Timothy L.\ Molloy%
\thanks{The authors are with the CIICADA Lab, School of Engineering, Australian National University (ANU), Canberra, ACT 2601, Australia (e-mail: \{tian.zhao1,timothy.molloy\}@anu.edu.au)}%
}

\maketitle
\thispagestyle{empty}
\pagestyle{empty}

\begin{abstract}
We formulate the discrete-time inverse optimal control problem of inferring unknown parameters in the objective function of an optimal control problem from measurements of optimal states and controls as a nonlinear filtering problem.
This formulation enables us to propose a novel extended Kalman filter (EKF) for solving inverse optimal control problems in a computationally efficient recursive online manner that requires only a single pass through the measurement data.
Importantly, we show that the Jacobians required to implement our EKF can be computed efficiently by exploiting recent Pontryagin differentiable programming results, and that our consideration of an EKF enables the development of first-of-their-kind theoretical error guarantees for online inverse optimal control with noisy incomplete measurements.
Our proposed EKF is shown to be significantly faster than an alternative unscented Kalman filter-based approach.
\end{abstract}

\newcommand\copyrighttext{%
  \footnotesize \textcopyright 2024 IEEE. Personal use of this material is permitted.
  Permission from IEEE must be obtained for all other uses, in any current or future
  media, including reprinting/republishing this material for advertising or promotional
  purposes, creating new collective works, for resale or redistribution to servers or
  lists, or reuse of any copyrighted component of this work in other works.}
\newcommand\copyrightnotice{%
\begin{tikzpicture}[remember picture,overlay]
\node[anchor=south,yshift=5pt] at (current page.south) {\fbox{\parbox{\dimexpr\textwidth-\fboxsep-\fboxrule\relax}{\copyrighttext}}};
\end{tikzpicture}%
}

\IEEEpeerreviewmaketitle
\copyrightnotice

\section{Introduction}

Inverse optimal control (IOC) techniques aim to infer unknown parameters in the objective function of optimal control problems from observed state and control data. 
IOC techniques have been widely developed and applied across various fields including control \cite{Keshavarz2011,Inga2021,awasthi_inverse_2020,Lian2021a,Lian2021,Maillot2013,Pauwels2014,Parsapour2021,Yu2021,Panchea2017,Zhang2018,Zhang2019}, machine learning \cite{Levine2012,jin_2021_pontryagin,Ng2000,Wulfmeier2015,Ziebart2008}, and robotics \cite{Mombaur2010,Aghasadeghi2014,Puydupin2012,Jin2018,Jin2019}.
Despite the range of applications and techniques developed, IOC remains a challenging problem, particularly in cases where the state and control data is incomplete, potentially noise corrupted, and must be processed in a single pass (i.e., online recursively without storing or reprocessing it in batches).
In this paper, we seek to develop a powerful novel approach to discrete-time IOC in such cases by posing and solving the challenging online IOC problem with noisy incomplete measurements as a nonlinear filtering problem.

Many existing approaches to discrete-time IOC have relied on a number of simplifying assumptions including that the state and control data is complete and not explicitly corrupted by (large) measurement noise.
Similarly, the objective function in many treatments of IOC is constructed as a linear combination of given basis functions with unknown parameters (see \cite[Chapter 3]{Molloy2022} and references therein).
Under such assumptions, IOC techniques have been developed to either find parameters that satisfy established optimality conditions such as the Karush-Kuhn-Tucker (KKT) conditions \cite{Puydupin2012,Keshavarz2011,Inga2021,Zhang2018,Zhang2019} or Pontryagin's principle \cite{Johnson2013,jin_2021_pontryagin,Molloy2022,Molloy2020}, or that explicitly optimize loss-functions based on errors between predicted and observed measurements in a bilevel-optimization manner \cite{Mombaur2010,Molloy2022}.
Furthermore, the vast majority of discrete-time IOC methods required multiple passes through the data, making them only suitable for inferring objective-function parameters offline (after all of the measurement data has been collected and stored).
Most discrete-time IOC methods are therefore unsuitable for practical real-time applications such as inferring the objectives of vehicles or users online from noisy sensor data and partial trajectories.




Recent work has addressed some aspects of handling either noisy and/or incomplete measurements in discrete-time IOC including \cite{Zhang2018,Zhang2019,Jin2019,Molloy2016},
However, these approaches require multiple passes through the data, and hence cannot be implemented in an efficient recursive online manner.
On the other hand, online discrete-time IOC methods that involve only a single pass through the data fail to explicitly handle measurement noise and data that may consist of only some states and/or control variables \cite{Molloy2018b,Molloy2020}.
In contrast, recent \emph{continuous-time} IOC approaches based on observers enable online continuous-time IOC with noisy incomplete measurements \cite{Self2019,Self2020,Self2020a,Self2021,Self2022,Lian2021,Lian2021a}.

Most recently, online approaches to inverse open-loop dynamic games, which essentially mirror online IOC approaches, have been studied in the context of autonomous vehicles \cite{Inga2021,cleach_2020_lucidgames}.
In particular, \cite{cleach_2020_lucidgames} proposed a novel approach based on the unscented Kalman filter (UKF). 
However, this UKF approach is computationally expensive and lacks theoretical error bounds (as well as not explicitly being developed for general incomplete measurements).

The key contributions of this paper are: 
\begin{enumerate}
    \item 
    The novel formulation of online discrete-time IOC with noisy incomplete measurements as a nonlinear filtering problem;
    \item
    The proposal of a new computationally efficient extended Kalman filter (EKF) for online discrete-time IOC with noisy incomplete measurements; and,
    \item
    The establishment of theoretical error guarantees for online discrete-time IOC with noisy incomplete measurements (using our EKF).
\end{enumerate}

The paper is structured as follows.
In Section \ref{sec:problem}, we formulate the online inverse optimal control problem with noisy incomplete measurements.
In Section \ref{sec:ekf}, we propose our EKF for online inverse optimal control, including how its required Jacobians can be computed efficiently.
In Section \ref{sec:error_analysis}, we establish conditions under which our proposed EKF will have bounded error.
In Section \ref{sec:results}, we examine the performance of our proposed EKF in simulations on several standard benchmark problems, and compare its computational efficiency to existing approaches. 
Finally, conclusions are presented in Section \ref{sec:conclusion}.

\paragraph*{Notation} The transpose of a matrix (or vector) $A$ will be denoted $A'$.
The identity matrix with dimensions $n \times n$ will be denoted $I_n$, and a matrix of all zeros with dimensions $n \times m$ will be denoted $0_{n \times m}$. 
Where the appropriate dimensions can be determined from context, we may also write $I$ or $0$ for the identity matrix and matrix/vector of zeros, respectively.
For matrices $A$ and $B$, the inequality $A \preceq B$ means that the difference $B - A$ is positive semidefinite.
We use $\| \cdot \|$ to denote the Euclidean $2$-norm (for vectors) or the matrix norm it induces, and $E [ \cdot ]$ to denote the expectation operator.

\section{Problem Formulation}
\label{sec:problem}

Consider a discrete-time deterministic system
\begin{equation}
    \label{eq:dynamics}
    \quad x_{t+1} = f\left({x}_t, {u}_t\right) \ \  \text{with given} \ {x}_0 \in \mathbb{R}^n
\end{equation}
for $t = 0, 1, \ldots, T$ where $0 < T < \infty$ is a given finite horizon, $x_t \in \mathbb{R}^n$ is the system state, $u_t \in \mathbb{R}^m$ is the control input, and $f: \mathbb{R}^n \times \mathbb{R}^m \rightarrow \mathbb{R}^n$ is a (potentially nonlinear) function describing the system dynamics, which we assume to be twice differentiable.
Let $x_{0:T}$ and $u_{0:T-1}$ denote state and control sequences $\{x_t : 0 \leq t \leq T\}$ and $\{u_t : 0 \leq t \leq T-1\}$, respectively.
Define the (parameterized) finite-horizon objective function
\begin{equation}
    \label{eq:objective_function}
    J(x_{0:T}, u_{0:T-1}, \theta)
    \triangleq c_T\left({x}_T, {\theta}\right) + \sum_{t=0}^{T-1} c_t\left({x}_t, {u}_t, {\theta}\right)
\end{equation}
where $\theta$ is a vector of parameters from the set $\Theta \subset \mathbb{R}^N$, and $c_t : \mathbb{R}^n \times \mathbb{R}^m \times \Theta \rightarrow \mathbb{R}$ and $c_T : \mathbb{R}^n \times \Theta \rightarrow \mathbb{R}$ are (parameterized) stage and terminal cost functions, respectively.
We assume that the cost functions $c_t$, and $c_T$ are twice differentiable.

In (standard or forward) discrete-time optimal control, the aim is to find state and control sequences that solve
\begin{align}
 \label{eq:ocp}
 \begin{aligned}
  &\underset{x_{0:T}, u_{0:T-1}}{{\inf}}  & & J(x_{0:T}, u_{0:T-1}, \theta)\\
  &\mathrm{s.t.} & & x_{t+1} = f (x_t, u_t), \quad t \geq 0 \\
  & & & x_{0} \in \mathbb{R}^n
 \end{aligned}
\end{align}
given the finite horizon $T$, the dynamics $f$, the time-invariant parameter vector $\theta$, and the functions $c_t$ and $c_T$. 
Let $x_{0:T}^\theta \triangleq \{x_t(\theta) : 0 \leq t \leq T\}$ and $u_{0:T-1}^\theta \triangleq \{u_t(\theta) : 0 \leq t \leq T - 1\}$ denote a sequence of states $x_t(\theta) \in \mathbb{R}^n$ and controls $u_t(\theta) \in \mathbb{R}^m$ that solve \eqref{eq:ocp} with given parameters $\theta$.

In this paper, we consider an \emph{online inverse optimal control} problem in which we aim to infer the parameters $\theta$ of the optimal control problem \eqref{eq:ocp} online from sequentially received noisy incomplete measurements of optimal states and controls $x_{0:T}^\theta = \{x_t(\theta) : 0 \leq t \leq T\}$ and $u_{0:T-1}^\theta = \{u_t(\theta) : 0 \leq t \leq T - 1\}$.
We assume that $F_t$ (and $g_t$) are known along with the initial state $x_0$, horizon $T$, the dynamics $f$, and the stage and terminal cost functions $c_t$ and $c_T$.
Specifically, at each $t = 0, 1, \ldots, T-1$, we receive a noisy incomplete measurement given by
\begin{align}
   \label{eq:measurements}
    y_t 
    &= F_t g_t(\theta) + v_t
\end{align}
where $v_t \sim N(0, R)$ is an independent and identically distributed Gaussian noise process (uncorrelated with the parameters $\theta$) with zero mean and covariance matrix $R \in \mathbb{R}^{q \times q}$.
Here, the functions $g_t : \mathbb{R}^N \rightarrow \mathbb{R}^{n + m}$ describe the implicit mapping of the parameters $\theta$ to the corresponding optimal states and controls at time $t$ through the solution of \eqref{eq:ocp}.
That is,
\begin{align*}
    g_t (\theta)
    &= \begin{bmatrix}
        x_t(\theta)\\
        u_t(\theta)
    \end{bmatrix}
\end{align*}
and, $F_t \in \mathbb{R}^{q \times (n + m)}$ is an arbitrary (potentially time-varying) matrix.
For example, if all the states and controls at time $t$ are measured then $F_t = I_{n + m}$; if only the states at time $t$ are measured then 
$ F_t = \begin{bmatrix}
        I_n & 0_{n \times m}
    \end{bmatrix}
$; and if only the controls at time $t$ are measured then $
F_t = \begin{bmatrix}
        0_{m \times n} & I_m
    \end{bmatrix}$.
We seek to use the measurements $\{ y_t : 0 \leq t \leq T-1\}$ sequentially online as they are received without storing or reprocessing them in batches (since storage and batch processing would delay inference).

Our online inverse optimal control problem differs from offline inverse optimal control problems in that the measurements must be processed sequentially online and not in batches (cf.~\cite[Section 3.7]{Molloy2022} and references therein).
It also differs from previous formulations of online inverse optimal control in discrete time (e.g., \cite{Molloy2018b,Molloy2020,Jin2018,cleach_2020_lucidgames}) by considering measurements corrupted by stochastic noise with partial state and/or control information.
Its consideration of stochastic noise and a finite horizon $T$ also differentiates it from previous online inverse optimal control formulations in continuous time (e.g., \cite{Self2019,Self2020,Self2020a,Self2021,Self2022,Lian2021,Lian2021a}).
Importantly, our formulation of online inverse optimal control will enable us to solve it as a nonlinear filtering problem using a novel computationally efficient recursive EKF.

\section{Proposed Extended Kalman Filter for Online Inverse Optimal Control}
\label{sec:ekf}

In this section, we develop our EKF for online inverse optimal control.

\subsection{Online Inverse Optimal Control as Nonlinear Filtering}

To develop the EKF, we first recast our online inverse optimal control problem as a nonlinear filtering problem.
Recalling that the unknown parameters $\theta \in \Theta$ to be inferred are time invariant and that the measurements $y_t$ are given by \eqref{eq:measurements}, the relationship between the parameters and the measurements are described by the nonlinear stochastic system
\begin{subequations}
\label{eq:ssModel}
\begin{align}
\label{eq:ss_dynamics}
    \theta_{t+1} 
    &= \theta_t\\\label{eq:ss_measurement}
    y_t 
    &= F_t g_t(\theta_t) + v_t
\end{align}
\end{subequations}
for $0 \leq t \leq T-1$ with initial conditions $\theta_0 = \theta \in \Theta$.
This system is nonlinear by virtue of the function $g_t$ describing the mapping of the parameters to optimal states and controls at time $t$.
With this nonlinear stochastic system description, the online inverse optimal control problem of inferring $\theta$ sequentially from the measurements $y_t$ constitutes a nonlinear filtering problem.

\subsection{EKF Algorithm}
We assign a Gaussian prior distribution $\theta \sim N(\hat{\theta}_0, P_0)$ with mean $\hat{\theta}_0 \in \mathbb{R}^N$ and covariance matrix $P_0 \in \mathbb{R}^{N \times N}$ to the unknown parameters $\theta$.

Due to the nonlinearity of the stochastic system \eqref{eq:ssModel}, we propose using an EKF to recursively compute an approximate Gaussian posterior distribution of the parameters $\theta$ at each time $t$ given current and previous measurements $y_{0:t}$.
This approximating Gaussian at time $t$ has mean $\hat{\theta}_t \in \mathbb{R}^N$ and covariance matrix $P_t \in \mathbb{R}^{N \times N}$.
Specialization of standard EKF equations (cf.\ \cite[Chapter 7]{sarkka2023bayesian}) to the system \eqref{eq:ssModel} implies that the mean $\hat{\theta}_t$ and covariance $P_t$ are given by the recursions
\begin{subequations}
\label{eq:ekf}
\begin{align}
\label{eq:ekf_predicted_covariance}
P_{t | t-1}
&= P_{t-1} + Q_t\\
K_t
&=P_{t|t-1} G_t' \left(G_t P_{t|t-1} G_t' + R \right)^{-1}\\\label{eq:ekf_update}
\hat{\theta}_t
&= \hat{\theta}_{t-1} + K_t \left(y_t - F_t g_t(\hat{\theta}_{t-1})\right) \\
P_t
&=P_{t | t-1} - K_t G_t P_{t | t-1}
\end{align}
\end{subequations}
for $1 \leq t \leq T-1$ with the initial filter mean and covariance taken from the prior $N(\hat{\theta}_0, P_0)$ where $P_{t | t-1} \in \mathbb{R}^{N \times N}$ are the predicted covariances; $Q_t \in \mathbb{R}^{N \times N}$ are \emph{user-tuneable} matrices (possibly equal to $0$); $K_t \in \mathbb{R}^{N \times q}$ are \emph{Kalman gains}; and $G_t \in \mathbb{R}^{q \times N}$ are the Jacobians of the deterministic part of the measurement equation \eqref{eq:ss_measurement} evaluated at $\hat{\theta}_{t-1}$, i.e.,
\begin{align}
\label{eq:jacobian}
 G_t
 \triangleq \dfrac{\partial }{\partial \theta}F_tg_t(\hat{\theta}_{t-1})
 = F_t \begin{bmatrix}
	\dfrac{\partial x_t}{\partial   {\theta}}(\hat{\theta}_{t-1})\\[0.25cm]
	\dfrac{\partial u_t}{\partial   {\theta}}(\hat{\theta}_{t-1})
 \end{bmatrix}.
 \end{align}
 
Our inclusion of the tuneable matrices $Q_t$ in our EKF \eqref{eq:ekf} is without stochastic justification because the associated nonlinear stochastic system \eqref{eq:ssModel} has no process noise (hence the matrices $Q_t$ are not the covariances).
Our inclusion of them is instead motivated by results showing that the consideration of user-selectable matrices $Q_t$ can offer practical as well as theoretical benefits to the performance of EKFs (cf.\ \cite[Section III]{reif1999stochastic} and references therein).
We shall discuss the selection of $Q_t$ in later sections analysing the error properties of our proposed EKF for online inverse optimal control.

The key challenge faced in implementing the EKF \eqref{eq:ekf} is computing the partial derivatives of the states $x_t(\theta)$ and controls $u_t(\theta)$ solving \eqref{eq:ocp} with $\theta = \hat{\theta}_{t-1}$ in the Jacobians \eqref{eq:jacobian}.
The perceived difficulty of computing such derivatives online at each time step $t$ has, in the past, led to EKF approaches to online inverse optimal control being ignored in favor of alternative filtering approaches that avoid Jacobian calculations, such as the UKF \cite{cleach_2020_lucidgames}.
To render our EKF tractable, we next exploit recent Pontryagin Differentiable Programming (PDP) results from \cite{jin_2021_pontryagin} to compute these derivatives.

\subsection{Jacobian via Pontryagin Differentiable Programming}

To compute the Jacobians in \eqref{eq:jacobian} using PDP \cite{jin_2021_pontryagin}, let us define the Hamiltonian function associated with the (parameterized) optimal control problem \eqref{eq:ocp} as
\begin{align*}
    H(t,x_t, u_t, p_{t+1}, \theta) 
    &\triangleq c_t(x_t, u_t, \theta) + f(x_t, u_t)' p_{t+1}
\end{align*}
where $p_{t} \in \mathbb{R}^n$ are costate (or adjoint) variables.
Since the stage-cost functions $c_t$, terminal-cost function $c_T$, and dynamics $f$  are twice continuously differentiable, the discrete-time Pontryagin's principle establishes that if the trajectories $x_{0:T}^\theta$ and $u_{0: T-1}^{\theta}$ solve \eqref{eq:ocp} then there exist a corresponding sequence of costates $p_{1:T}^\theta = \{ p_t(\theta) : 1 \leq t \leq T \}$ satisfying
\begin{align}
\label{eq:costate}
\begin{split}
 p_t(\theta)
 &= \frac{\partial  c_t}{\partial  {x}_t}(x_t(\theta), u_t(\theta), \theta) \\
 &\quad+ {\frac{\partial  {f}'}{\partial  {x}_t} }(x_t(\theta), u_t(\theta))  { p}_{t+1}(\theta) 
\end{split}
\end{align}
for $1 \leq t \leq T-1$ with $p_T(\theta) = \frac{\partial c_T}{\partial x_T}(x_t(\theta), \theta)$.

Recall that in our EKF \eqref{eq:ekf} we are interested in computing the Jacobian $G_t$ satisfying \eqref{eq:jacobian} at time $t$ given the parameter estimate $\hat{\theta}_{t-1}$ from the previous time step $t-1$.
By solving the (forward) discrete-time optimal control problem \eqref{eq:ocp} with $\theta = \hat{\theta}_{t-1}$ and iterating the costate equation \eqref{eq:costate}, we generate \emph{predicted optimal sequences} $x_{0:T}^{\hat{\theta}_{t-1}} = \{ x_k(\hat{\theta}_{t-1}) : 0 \leq k \leq T\}$, $u_{0:T-1}^{\hat{\theta}_{t-1}} = \{ u_k(\hat{\theta}_{t-1}) : 0 \leq k \leq T-1\}$, and $p_{1:T}^{\hat{\theta}_{t-1}} = \{ p_k(\hat{\theta}_{t-1}) : 1 \leq k \leq T\}$.
Here we use $0 \leq k \leq T$ as the time index to highlight that these sequences are not occurring concurrently with the (true) time steps of the filter (i.e., $t$ and $\hat{\theta}_{t-1}$ are fixed, and we varying the time index $k$ to examine states, controls, and costates at different points in the horizon assuming that $\theta = \hat{\theta}_{t-1}$ in \eqref{eq:ocp}).

Given the predicted optimal sequences $x_{0:T}^{\hat{\theta}_{t-1}}$, $u_{0:T-1}^{\hat{\theta}_{t-1}}$, and $p_{1:T}^{\hat{\theta}_{t-1}}$ and the time index $k$ seperate from the time on which the EKF is operating, we can now follow \cite{jin_2021_pontryagin} by defining the following Jacobians and Hessians:
\begin{equation}
\label{eq:jacobianAndHessians}
\begin{aligned}
A_k 
&\triangleq \frac{\partial  {f}}{\partial  x_t} (x_k(\hat{\theta}_{t-1}), u_k(\hat{\theta}_{t-1}))\\
B_k 
&\triangleq \frac{\partial  {f}}{\partial  u_t} (x_k(\hat{\theta}_{t-1}), u_k(\hat{\theta}_{t-1})) \\ 
H_k^{x x}
&\triangleq \frac{\partial^2 H}{\partial  x_t \partial  x_t} (k, x_k(\hat{\theta}_{t-1}), u_k(\hat{\theta}_{t-1}), p_{k+1}(\hat{\theta}_{t-1}), \hat{\theta}_{t-1}) \\
H_k^{x e} 
&\triangleq \frac{\partial^2 H}{\partial  {x}_t \partial  {\theta}} (k, x_k(\hat{\theta}_{t-1}), u_k(\hat{\theta}_{t-1}), p_{k+1}(\hat{\theta}_{t-1}), \hat{\theta}_{t-1}) \\
H_k^{x u} 
&\triangleq \frac{\partial^2 H}{\partial  {x}_t \partial  {u}_t} (k, x_k(\hat{\theta}_{t-1}), u_k(\hat{\theta}_{t-1}), p_{k+1}(\hat{\theta}_{t-1}), \hat{\theta}_{t-1}) \\
&\triangleq \left(H_k^{u x}\right)^{\prime}, \\
H_k^{u u} 
&\triangleq \frac{\partial^2 H}{\partial  {u}_t \partial  {u}_t}(k, x_k(\hat{\theta}_{t-1}), u_k(\hat{\theta}_{t-1}), p_{k+1}(\hat{\theta}_{t-1}), \hat{\theta}_{t-1}) \\
H_k^{u e} 
&\triangleq \frac{\partial^2 H}{\partial  {u}_t \partial  {\theta}}, (k, x_k(\hat{\theta}_{t-1}), u_k(\hat{\theta}_{t-1}), p_{k+1}(\hat{\theta}_{t-1}), \hat{\theta}_{t-1}) \\
H_T^{x x} 
&\triangleq \frac{\partial^2 c_T}{\partial  {x}_T \partial  {x}_T}(x_T(\hat{\theta}_{t-1}), \hat{\theta}_{t-1}) \\
H_T^{x e} 
&\triangleq \frac{\partial^2 c_T}{\partial  {x}_T \partial  {\theta}} (x_T(\hat{\theta}_{t-1}), \hat{\theta}_{t-1}).
\end{aligned}
\end{equation}
Let us also define new ``state'' and ``control'' variables representing the unknown partial derivatives of the optimal states and controls evaluated at $\hat{\theta}_{t-1}$, namely,
\begin{equation}
    X_k(\hat{\theta}_{t-1}) \triangleq \frac{\partial x_k}{\partial   {\theta}}(\hat{\theta}_{t-1}) \text{ and } U_k(\hat{\theta}_{t-1}) \triangleq \frac{\partial u_k}{\partial {\theta}} (\hat{\theta}_{t-1}).
\end{equation}
Note that by setting $k = t$, the partial derivatives $X_t(\hat{\theta}_{t-1})$ and $U_t(\hat{\theta}_{t-1})$ correspond to those required to compute the Jacobian $G_t$ satisfying \eqref{eq:jacobian} at time $t$ in our EKF \eqref{eq:ekf}.
The following proposition, based on the results of \cite{jin_2021_pontryagin}, establishes efficient recursions for computing these derivatives.

\begin{proposition}
\label{proposition:derivatives}
Suppose that $H_k^{uu}$ is invertible for all $0 \leq k \leq T-1$.
Then $X_{0:T}^{\hat{\theta}_{t-1}} \triangleq \{X_k(\hat{\theta}_{t-1}) : 0 \leq k \leq T \}$ and $U_{0:T-1}^{\hat{\theta}_{t-1}} \triangleq \{ U_k(\hat{\theta}_{t-1}) : 0 \leq k \leq T-1 \}$ can be obtained via the recursions:
\begin{subequations}
\label{eq:firstRecursions}
\begin{align}
  \mathcal{P}_k
  &= \mathcal{Q}_k + \mathcal{A}_k'(I + \mathcal{P}_{k+1} \mathcal{R}_k)^{-1}\mathcal{P}_{k+1}\mathcal{A}_k\\
  \mathcal{W}_k
  &= \mathcal{Q}_k'(I + \mathcal{P}_{k+1}\mathcal{R}_k)^{-1}(\mathcal{W}_{k+1}+\mathcal{P}_{k+1}\mathcal{M}_k) + \mathcal{N}_k
\end{align}
\end{subequations}
for $0 \leq k \leq T-1$ where $\mathcal{P}_T = H_T^{xx}$ and $\mathcal{W}_T = H_T^{xe}$, together with
\begin{subequations}
\label{eq:secondRecursions}
\begin{align}\notag
  &U_k(\hat{\theta}_{t-1})\\\notag
  &= -(H_k^{uu})^{-1}(H_k^{ux}X_k(\hat{\theta}_{t-1}) + H_k^{ue} \\
  &\quad+ B_k'(I + \mathcal{P}_{k+1}\mathcal{R}_k)^{-1})^{-1}(\mathcal{P}_{k+1} \mathcal{A}_{k+1}  X_k(\hat{\theta}_{t-1}) \\\notag
  &\quad+ \mathcal{P}_{k+1} \mathcal{M}_k + \mathcal{W}_{k+1})\\
  &X_{k+1}(\hat{\theta}_{t-1}) = A_k X_k(\hat{\theta}_{t-1}) + B_k U_k(\hat{\theta}_{t-1})
\end{align}
\end{subequations}
for $0 \leq k \leq T$ with $X_0(\hat{\theta}_{t-1}) = 0$ and where $\mathcal{A}_k = A_k - B_k(H_k^{uu})^{-1}H_k^{ux}$, $\mathcal{R}_k = B_k(H_k^{uu})^{-1} B_k'$, $M_k = -B_k(H_k^{uu})^{-1}H_k^{ue}$, $\mathcal{Q}_k = H_k^{xx} - H_k^{xu}(H_k^{uu})^{-1}H_k^{ux}$, $\mathcal{N}_k = H_k^{xe} - H_k^{xu}(H_k^{uu})^{-1}H_k^{ue}$ can be computed via \eqref{eq:jacobianAndHessians}  by solving \eqref{eq:ocp} with $\theta = \hat{\theta}_{t-1}$ and iterating \eqref{eq:costate}.
\end{proposition}
\begin{proof}
Follows from Lemmas 5.1 and 5.2 of \cite{jin_2021_pontryagin}.
\end{proof}

In light of Proposition \ref{proposition:derivatives}, the partial derivatives $X_t(\hat{\theta}_{t-1})$ and $U_t(\hat{\theta}_{t-1})$, and hence the Jacobian $G_t$, at time $t$ in our EKF \eqref{eq:ekf} can be computed efficiently by: 1) solving \eqref{eq:ocp} with $\theta = \hat{\theta}_{t-1}$ and iterating \eqref{eq:costate}; 2) computing the Jacobians and Hessians in \eqref{eq:jacobianAndHessians}; and, 3) solving the recursions \eqref{eq:firstRecursions} and \eqref{eq:secondRecursions} until $k=t$ giving $X_t(\hat{\theta}_{t-1})$ and $U_t(\hat{\theta}_{t-1})$. 
We highlight that the recursions \eqref{eq:firstRecursions} and \eqref{eq:secondRecursions} essentially constitute the same routine as solving a linear-quadratic optimal control problem (see \cite{jin_2021_pontryagin} for more details), and thus the calculation of the Jacobian $G_t$ at time $t$ in our EKF \eqref{eq:ekf} requires the solution of only two optimal control problems --- the original optimal control problem \eqref{eq:ocp} and the linear-quadratic-like problem implicitly associated with the recursions in Proposition \ref{proposition:derivatives}.
The need to solve only two optimal control problems at each time $t$ (with only one being the original) compares surprisingly favourably with previous approaches.
For example, the UKF approach proposed in \cite{cleach_2020_lucidgames} requires the original optimal control problem \eqref{eq:ocp} to be solved multiple times at each time $t$, specifically, at each time $t$, it requires \eqref{eq:ocp} to be solved once for each sigma-point, with the number of sigma-points in a UKF usually taken to be $2N + 1$ (where $N$ is the dimension of the unknown parameters $\theta$).

\subsection{Proposed Algorithm Summary}

Our proposed EKF for online inverse optimal control is summarized in Algorithm \ref{alg:ekf}.
The algorithm for computing the Jacobian at time $t$ is summarized in Algorithm \ref{alg:jaco}.
The system dynamics $f$, cost functions $\{c_t : 0 \leq t \leq T\}$, initial state $x_0$, and horizon $T$ are given (although the parameters $\theta$ of the cost functions are unknown). 
The output of the EKF at each time $t$ is the estimated unknown parameters $\hat{\theta}_t$ and the associated covariance $P_t$, providing a recursive solution to online inverse optimal control with noisy incomplete measurements.

\let\oldReturn\Return
\renewcommand{\Return}{\State\oldReturn}

\begin{algorithm}[t!]
\caption{Proposed Extended Kalman Filter for IOC}
\label{alg:ekf}
\begin{algorithmic}[1]
\Procedure{EKF}{$\hat{\theta}_0$, $P_0$, $\{Q_t : 0 \leq t < T\}$, $R$, $T$, $\{F_t : 0 \leq k \leq T\}$, $f$, $\{c_t : 0 \leq t \leq T\}$}
    \For {$t = 1,..,T-1$}
    
    \State // Compute Predicted Covariance
    \State $P_{t|t-1} \leftarrow P_{t-1} + Q_t$  
    
    \State // Solve forward optimal control for trajectories
    \State $x_{0:T}^{\hat{\theta}_{t-1}}$, $u_{0:T-1}^{\hat{\theta}_{t-1}} \leftarrow$ via solving \eqref{eq:ocp} with $\theta = \hat{\theta}_{t-1}$.
    \State // Compute Jacobian
    \State $G_t \leftarrow \textsc{jacobian} (t, \hat{\theta}_{t-1}, x_{0:T}^{\hat{\theta}_{t-1}},\cdots$\\
    $\quad \qquad \qquad \qquad\qquad u_{0:T-1}^{\hat{\theta}_{t-1}}, F_{t}, f, \{c_t : 0 \leq t \leq T\})$
    \State // Compute Kalman Gain
    \State $K_t \leftarrow P_{t|t-1} G_t' \left(G_t P_{t|t-1} G_t' + R \right)^{-1}$ 

    \State // New measurement at time $t$
    \State Receive noisy measurement $y_t$   \Comment{Eq. \ref{eq:measurements}}

    \State // Update Parameter Mean
    \State $\hat{\theta}_t \leftarrow \hat{\theta}_{t-1} + K_t \left(y_t - F_t g_t(\hat{\theta}_{t-1})\right)$
    \State // Update Covariance
    \State $P_t \leftarrow P_{t | t-1} - K_t G_t P_{t | t-1}$
    \EndFor
\EndProcedure

\end{algorithmic}
\end{algorithm}

\begin{algorithm}[t!]
\caption{Jacobian Computation}
\label{alg:jaco}
\begin{algorithmic}[1]

\Procedure{Jacobian}{$t, \hat{\theta}_{t-1}, x_{0:T}^{\hat{\theta}_{t-1}}, u_{0:T-1}^{\hat{\theta}_{t-1}}, \ F_t, f, \{c_t : 0 \leq t \leq T\}$}      
    

    

    \State // Compute costates
    \State $p_{1:T}^{\hat{\theta}_{t-1}} \leftarrow$ via \eqref{eq:costate} using $\{c_t : 0 \leq t \leq T\}$ and $f$
    \State // Compute Jacobians and Hessians
    \State Evaluate \eqref{eq:jacobianAndHessians} using $\{c_t : 0 \leq t \leq T\}$ and $f$
    \State // Solve linear-quadratic optimal control problem
    \State $X_{0:T}^{\hat{\theta}_{t-1}}$,$U_{0:T-1}^{\hat{\theta}_{t-1}} \leftarrow $ via \eqref{eq:firstRecursions} \& \eqref{eq:secondRecursions} \Comment{Proposition \ref{proposition:derivatives}}
    \State  // Select Jacobian for available measurements
    \State $ G_t \leftarrow F_t[X_{t}({\hat{\theta}_{t-1}})' \quad U_{t}({\hat{\theta}_{t-1}})']' $

    
    \Return{$G_t$}
\EndProcedure
\end{algorithmic}
\end{algorithm}

\section{Error Analysis and Guarantee}
\label{sec:error_analysis}

In this section, we establish conditions under which our proposed EKF \eqref{eq:ekf} will have bounded error in the sense of the following definition.

\begin{definition}[{Mean-Squared Boundedness \cite{reif1999stochastic}}]
\label{definition:msb}
 A stochastic process $\{\zeta_t \in \mathbb{R}^N : t \geq 0\}$ is said to be (exponentially) \emph{mean-squared bounded} if there exist real numbers $\eta, \nu > 0$ and $0 < \vartheta < 1$ such that
 \begin{align*}
     E[\| \zeta_t \|^2] \leq \eta \| \zeta_0 \|^2 \vartheta^t + \nu
 \end{align*}
 holds for all $t \geq 0$.
\end{definition}

To perform our error analysis, let us define the estimation error of our EKF \eqref{eq:ekf} at time $t$ as 
\begin{align}
    \label{eq:estimationError}
    e_t 
    &\triangleq \theta - \hat{\theta}_{t-1}.
\end{align}
Let us also note that in light of the results of the preceding section, the deterministic part of the measurement equation \eqref{eq:ss_measurement} has a Taylor series expansion around the estimate $\hat{\theta}_{t-1}$ given by 
\begin{align}
    \label{eq:taylor}
    F_tg_t(\theta) - F_t g_t(\hat{\theta}_{t-1})
    &= G_t e_t + \chi(\theta, \hat{\theta}_{t-1})
\end{align}
where $\chi ( \cdot, \cdot)$ is a suitable nonlinear function that accounts for the higher-order terms of the expansion.
The following proposition establishes (sufficient) conditions under which the error $e_t$ of our EKF \eqref{eq:ekf} is mean-squared bounded in the sense of Definition \ref{definition:msb}.

\begin{proposition}
    \label{proposition:msb}
    Consider our EKF \eqref{eq:ekf} for the nonlinear stochastic system model of inverse optimal control in \eqref{eq:ssModel} and suppose that:
    \begin{enumerate}
        \item
        There exist positive reals $\bar{g},\underline{p},\bar{p}, \underline{q}, \underline{r}, \delta >0$ such that
        \begin{align}\label{eq:gCondition}
            \| G_t \| &\leq \bar{g}\\\label{eq:pCondition}
            \underline{p}I &\preceq P_{t | t-1} \preceq \bar{p}I\\\label{eq:qCondition}
            \underline{q} I &\preceq Q_t \preceq \delta I \\
            \underline{r} I &\preceq R \preceq \delta I
        \end{align}
        for all $0 \leq t \leq T-1$; and,
        \item 
        There exist positive reals $\epsilon_\chi, \kappa_\chi > 0$ such that the nonlinear function $\chi$ satisfies
        \begin{align}
            \| \chi (\theta, \hat{\theta}) \| \leq \kappa_\chi \| \theta - \hat{\theta}\|^2 \; \mathrm{ for } \; \| \theta - \hat{\theta}\| \leq \epsilon_\chi.
        \end{align}
    \end{enumerate}
    Then, the estimation error $e_t$ given by \eqref{eq:estimationError} is mean-squared bounded in the sense of Definition \ref{definition:msb} provided that the initial estimation error is bounded, i.e., provided that $\| \theta - \hat{\theta}_0 \| \leq \epsilon$ for some $\epsilon > 0$.
\end{proposition}
\begin{proof}
    Follows from \cite[Theorem 3.1]{reif1999stochastic}, with the constant dynamics in \eqref{eq:ss_dynamics} implying that several of its conditions hold automatically (e.g., \cite[Eq.\ (28) \& (33)]{reif1999stochastic}).
\end{proof}

Proposition \ref{proposition:msb} is novel in the context of online inverse optimal control since it provides the first error bounds (and insight into convergence rates) for a discrete-time online inverse optimal control method.
Previous results in \cite{Molloy2018b,Molloy2020} have only established conditions under which the solutions to discrete-time inverse optimal control problems are unique (without consideration of the error in these unique estimates).
Combining our new error-analysis results with these existing uniqueness results is a focus of our current research efforts.

Proposition \ref{proposition:msb} is also of considerable practical importance because its conditions provide insight into understanding when our EKF \eqref{eq:ekf} will yield accurate parameter estimates.
In particular, the condition on the tuneable matrices $Q_t$ in \eqref{eq:qCondition} suggests that they should be selected such that they are bounded away from $0$.
The conditions on the Jacobians $G_t$ in \eqref{eq:gCondition} and covariance matrices $P_{t | t-1}$ in \eqref{eq:pCondition} are analogous to observability (of the linearized system) or persistence of excitation conditions since they can be verified given specific measurement data to determine if the estimates produced by our EKF \eqref{eq:ekf} are reliable or not (see also the discussion of these conditions in \cite[Section IV]{reif1999stochastic}).
Similar conditions have been developed in the context of both offline \cite{Molloy2018,Molloy2022} and online \cite{Molloy2018b,Molloy2020} inverse optimal control before, these previous conditions only hold in the case of noise-free complete measurements of the form $y_t = \begin{bmatrix} x_t(\theta)' & u_t(\theta)'\end{bmatrix}'$, and not in the case of noisy incomplete measurements that could arise from the more general measurement model in \eqref{eq:ss_measurement}.

\section{Simulation Results}
\label{sec:results}

In this section, we utilise the benchmark problems presented and implemented in \cite{jin_2021_pontryagin} to illustrate and compare the practical performance of our algorithm. Additionally, we evaluate its computational efficiency compared to a UKF approach based on that proposed in \cite{cleach_2020_lucidgames}. 

\subsection{Benchmark Problems}

\begin{table}[t!]
\centering
\caption{Benchmark Problems Dimensions and Measurement Noise Covariance Diagonal (Noise Magnitude)} 

\begin{tabular}{cccc}
\hline
Problem & $q = m + n$  &  $N$ & Noise Magnitude (diag of R) \\
\hline

Single pendulum & 3 & 2 & \num{1e-7} \\

Cart pole & 5 &  4 & \num{1e-6}\\

Robot arm & 6 &  4 & \num{1e-5}\\

Quadrotor & 17 & 4 &  \num{1e-7} \\

Rocket landing & 16 & 5 & \num{1e-6} \\

\hline
\end{tabular}
\end{table}

To illustrate our proposed EKF algorithm for IOC with noisy measurements but complete state and control information, we simulated the single pendulum, cart pole, quadrotor, robot arm and rocket powered landing benchmark problems from \cite{jin_2021_pontryagin}. 
For each benchmark problem, the system dynamics are nonlinear and deterministic, with unknown parameters $\theta$ in the objective function to be inferred from simulated states and controls. 
The total number of states and controls $n + m$ in each example, which corresponds to the measurement dimensions $q$ in this complete-information case ($F_t = I_{n + m}$), and the number of unknown parameters $N$ are detailed in Table I.


\begin{figure*}[htbp]
\label{fig:benchmark_problems}
\centering
\subfigure[Single pendulum]{
\includegraphics[width=0.7\columnwidth]{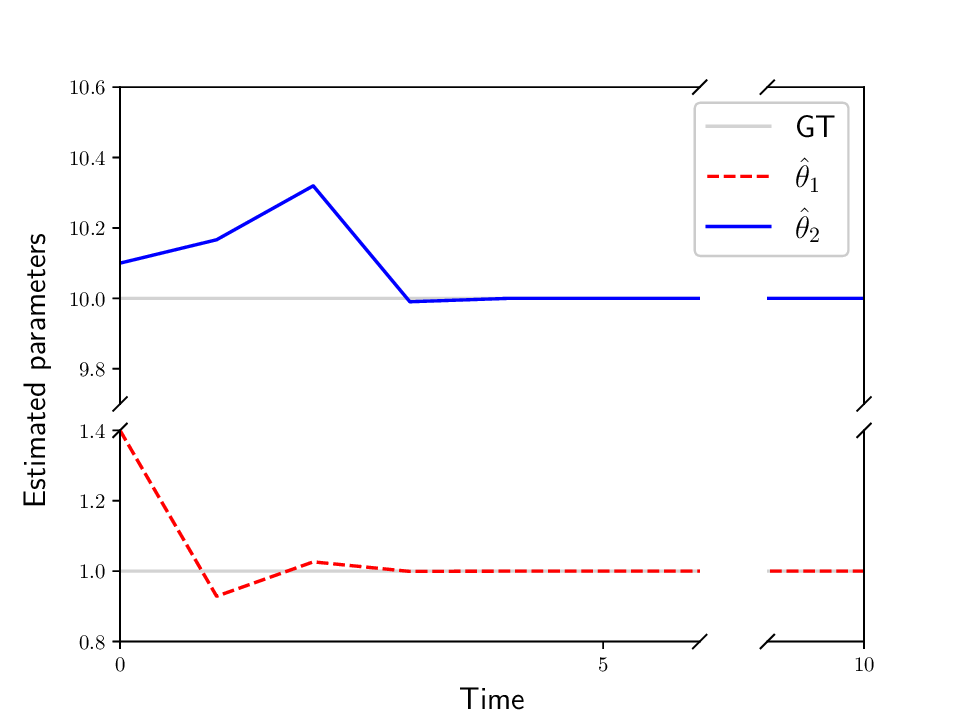}
}
\hspace{-8mm}
\subfigure[Cart pole]{
\includegraphics[width=0.7\columnwidth]{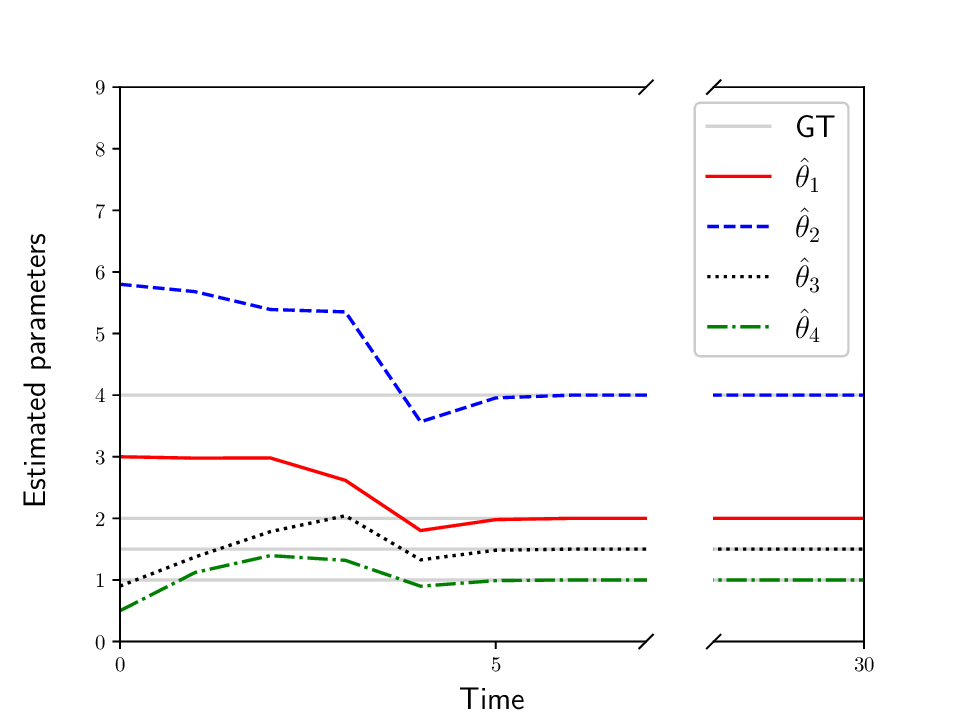}
}
\hspace{-8mm}
\subfigure[Quadrotor]{
\includegraphics[width=0.7\columnwidth]{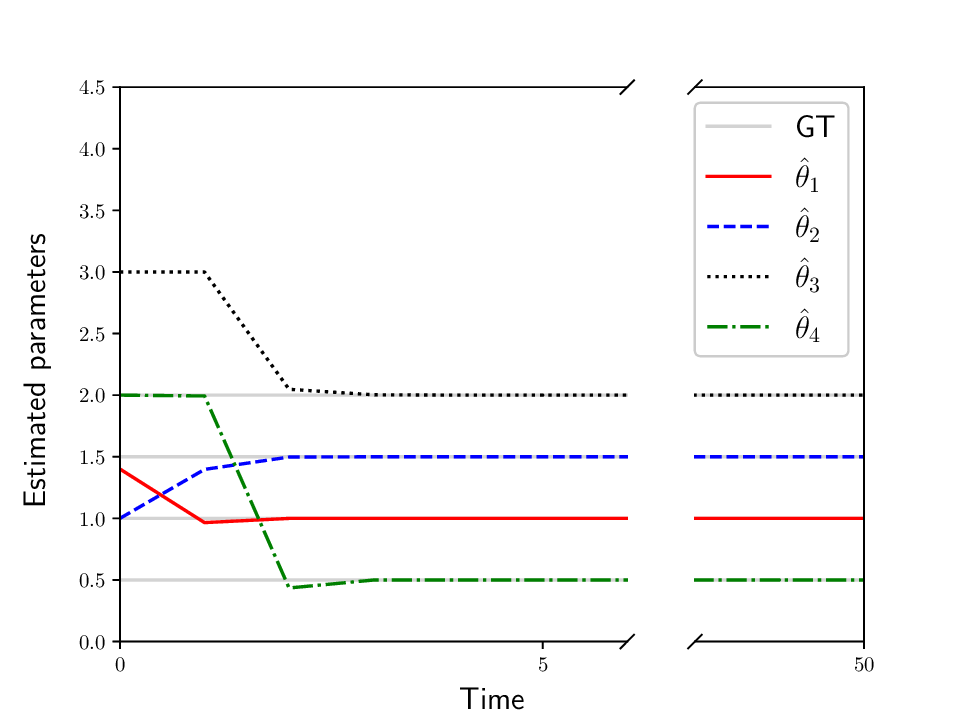}
}
\hspace{-8mm}
\subfigure[Robot arm]{
\includegraphics[width=0.7\columnwidth]{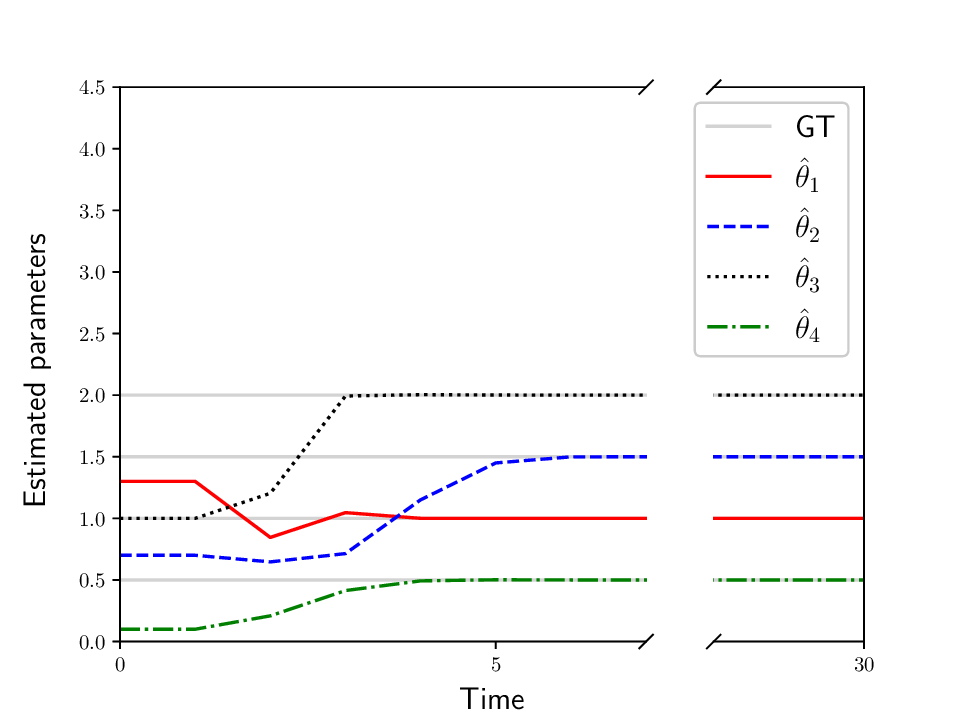}
}
\hspace{-8mm}
\subfigure[Rocket powered landing]{
\includegraphics[width=0.7\columnwidth]{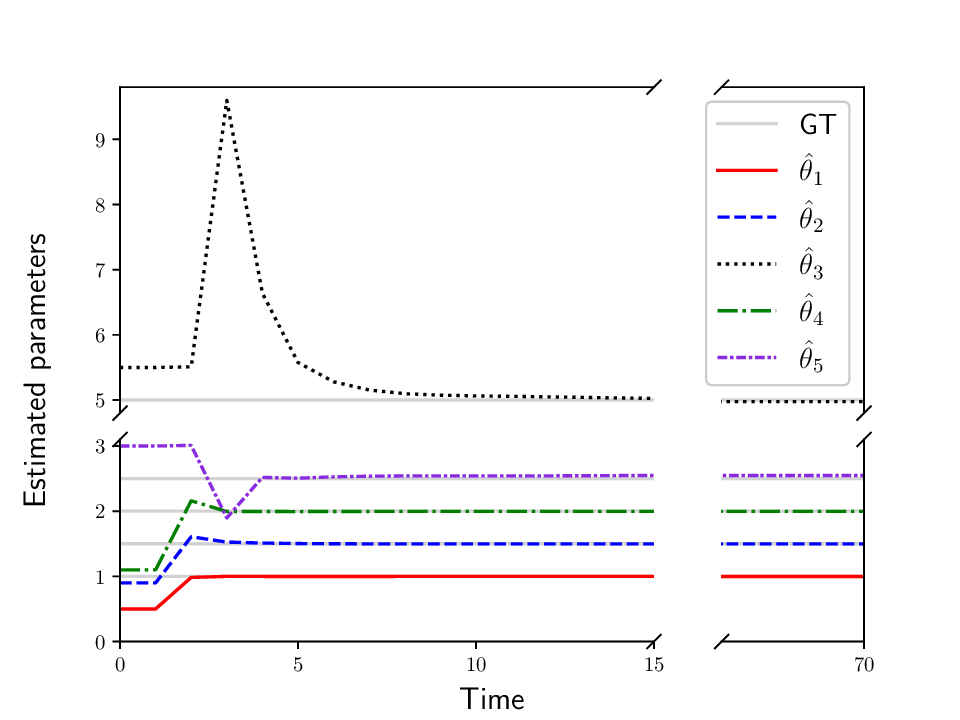}
}
\hspace{-8mm}
\subfigure[Execution time comparison]{
\includegraphics[width=0.7\columnwidth]{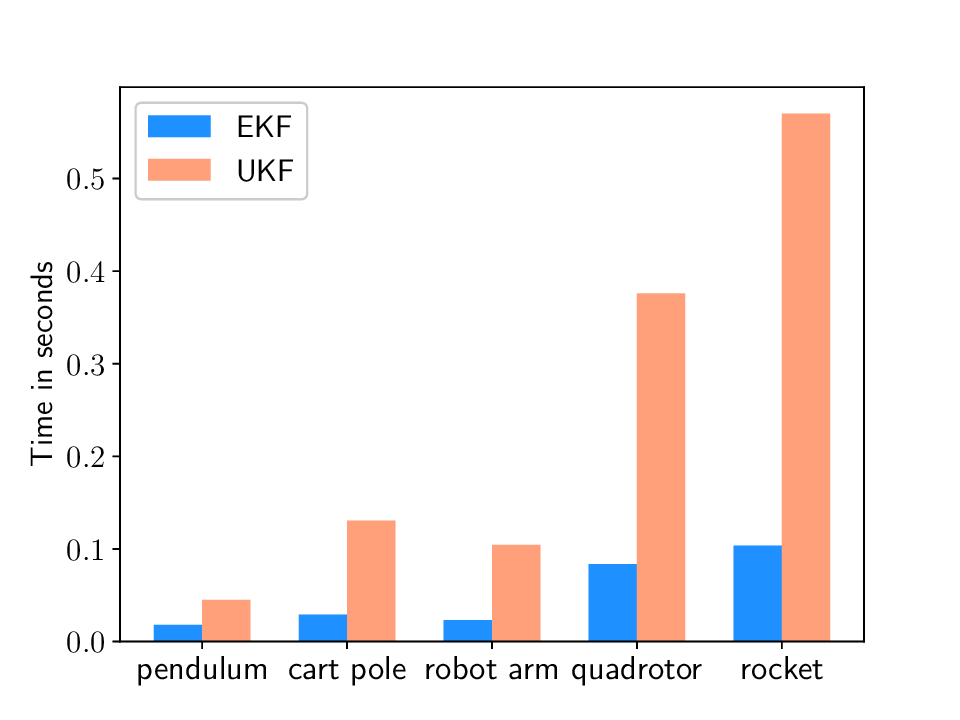}
}
\caption{Performance of our proposed EKF on benchmark problems with ground truth (GT): (a) Single pendulum with ground truth parameters $[\theta_1 = 1, \theta_2 = 10]$; (b) Cart pole with ground truth parameters $[\theta_1 = 2, \theta_2 = 4, \theta_3 = 1.5, \theta_4 = 1]$; (c) Quadrotor with ground truth parameters $[\theta_1=1.0, \theta_2 = 1.5, \theta_3 = 2, \theta_4 = 0.5]$; (d) Robot arm with ground truth parameters $[\theta_1=1.0, \theta_2 = 1.5, \theta_3 = 2, \theta_4 = 0.5]$; (e) Rocket powered landing with ground truth parameters $[\theta_1=1.0, \theta_2 = 1.5, \theta_3 = 2, \theta_4 = 2.5, \theta_5 =5 ]$; (f) Execution time for each time step with respect to EKF and UKF.}
\end{figure*}

\subsubsection{Noise Simulations}
For each benchmark problem, we generated measurements by adding Gaussian measurement noise $v_t$ to the states and controls.
The measurement noise covariance matrices $R$ were diagonal with the diagonal elements given in Table I.
We applied our proposed EKF algorithm to these simulated trajectories and the resulting estimates are shown in (a) - (e) of Fig. 1. 
The results show that our proposed EKF algorithm converges to the true parameters of the objective function (with negligible error). 
We note that even in this case of relatively small measurement noise, the existing recursive online discrete-time inverse optimal control approach of \cite{Molloy2020} is known to perform poorly (and hence its performance is not reported).
The UKF-based approach of \cite{cleach_2020_lucidgames} yielded similarly accurate parameter estimates to our proposed EKF.

\subsubsection{Computational Efficiency}
We also recorded the computational time at each time step $t$ required by both our proposed EKF and the UKF-based approach of \cite{cleach_2020_lucidgames}.
Fig. 1 (f) reports these times for each benchmark problem conducted on a 10-core processor (Apple M2).
We see that our derivative-based EKF requires significantly less time than the UKF-based approach. 
The great advantage that our EKF has over the UKF in terms of computational efficiency is because for each time step of estimation process, our EKF only involves solving $2$ optimal control problems (with one being a linear-quadratic problem), whilst the UKF-based approach requires the solution of an optimal control problem per sigma point.
As is standard in UKFs \cite{cleach_2020_lucidgames}, we selected $2N + 1$ sigma points, and therefore had to solve $2N + 1$ optimal control problems per time step.

\subsection{Complete versus Incomplete Measurement Simulation}
To examine the performance of our proposed EKF in the case of incomplete measurements, we simulated the single pendulum problem with both: 1) complete measurements of all states and controls ($F_t = I_{n + m}$); and, 2) incomplete measurements with only states ($ F_t = \begin{bmatrix}
        I_n & 0_{n \times m}
    \end{bmatrix}
$). 
The results are shown in Fig. 2.
We see that in this case our proposed EKF is able to solve the IOC problem in both cases, but its convergence speed is slightly slower.
We note that the existing recursive online discrete-time inverse optimal control approach of \cite{Molloy2020} requires complete state and control information, and hence cannot estimate the parameters from these incomplete measurements.

\begin{figure}[t!]
    \centering    
    \includegraphics[width=0.98\columnwidth]{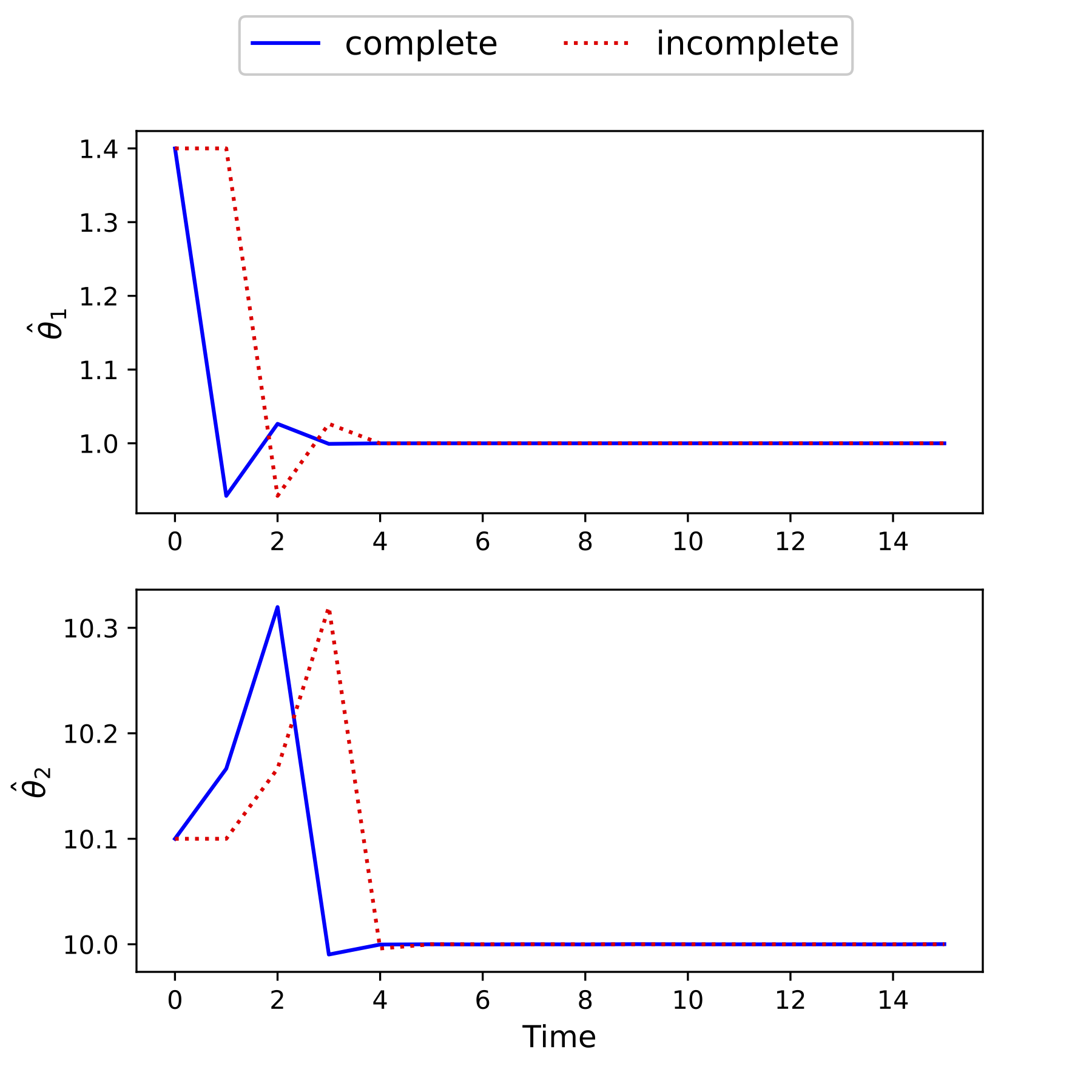}
    \caption{Estimated parameters from our proposed EKF with complete and incomplete measurements for solving the single pendulum benchmark problem.}
    \label{fig:enter-label}
\end{figure}

\section{Conclusion}
\label{sec:conclusion}
We posed the problem of online inverse optimal control with imperfect measurements as a nonlinear filtering problem, and proposed a computationally efficient extended Kalman filter (EKF).
Our EKF requires only a single pass through the data, involves the solution of at most two optimal control problems per time step, and is shown to offer provably bounded mean-squared error under mild conditions.
In contrast, existing approaches to inverse optimal control require combinations of multiple passes through the data, the solution of multiple optimal control problems per time step, and/or fail with unbounded error (in both theory and practice) if the data is imperfect.
We illustrated the efficiency and performance of our proposed EKF on several standard benchmark problems.

\bibliographystyle{IEEEtran}
\bibliography{references}  

\begin{thebibliography}{10}
\providecommand{\url}[1]{#1}
\csname url@samestyle\endcsname
\providecommand{\newblock}{\relax}
\providecommand{\bibinfo}[2]{#2}
\providecommand{\BIBentrySTDinterwordspacing}{\spaceskip=0pt\relax}
\providecommand{\BIBentryALTinterwordstretchfactor}{4}
\providecommand{\BIBentryALTinterwordspacing}{\spaceskip=\fontdimen2\font plus
\BIBentryALTinterwordstretchfactor\fontdimen3\font minus
  \fontdimen4\font\relax}
\providecommand{\BIBforeignlanguage}[2]{{%
\expandafter\ifx\csname l@#1\endcsname\relax
\typeout{** WARNING: IEEEtran.bst: No hyphenation pattern has been}%
\typeout{** loaded for the language `#1'. Using the pattern for}%
\typeout{** the default language instead.}%
\else
\language=\csname l@#1\endcsname
\fi
#2}}
\providecommand{\BIBdecl}{\relax}
\BIBdecl

\bibitem{Keshavarz2011}
A.~Keshavarz, Y.~Wang, and S.~Boyd, ``Imputing a convex objective function,''
  in \emph{Intelligent Control (ISIC), 2011 IEEE International Symposium
  on}.\hskip 1em plus 0.5em minus 0.4em\relax IEEE, 2011, pp. 613--619.

\bibitem{Inga2021}
J.~Inga, A.~Creutz, and S.~Hohmann, ``{Online Inverse Linear-Quadratic
  Differential Games Applied to Human Behavior Identification in Shared
  Control},'' in \emph{2021 {European} {Control} {Conference} ({ECC})}, 2021.

\bibitem{awasthi_inverse_2020}
C.~Awasthi and A.~Lamperski, ``Inverse {Differential} {Games} {With} {Mixed}
  {Inequality} {Constraints},'' in \emph{2020 {American} {Control} {Conference}
  ({ACC})}, Jul. 2020, pp. 2182--2187, iSSN: 2378-5861.

\bibitem{Lian2021a}
B.~Lian, W.~Xue, F.~L. Lewis, and T.~Chai, ``Online inverse reinforcement
  learning for nonlinear systems with adversarial attacks,''
  \emph{International Journal of Robust and Nonlinear Control}, vol.~31,
  no.~14, pp. 6646--6667, 2021.

\bibitem{Lian2021}
------, ``{Robust Inverse Q-Learning for Continuous-Time Linear Systems in
  Adversarial Environments},'' \emph{IEEE Transactions on Cybernetics}, pp.
  1--13, 2021.

\bibitem{Maillot2013}
T.~Maillot, U.~Serres, J.-P. Gauthier, and A.~Ajami, ``How pilots fly: an
  inverse optimal control problem approach,'' in \emph{Decision and Control
  (CDC), 2013 IEEE 52nd Annual Conference on}.\hskip 1em plus 0.5em minus
  0.4em\relax IEEE, 2013, pp. 1792--1797.

\bibitem{Pauwels2014}
E.~Pauwels, D.~Henrion, and J.-B. Lasserre, ``Inverse optimal control with
  polynomial optimization,'' in \emph{Decision and Control (CDC), 2014 IEEE
  53rd Annual Conference on}, Dec 2014, pp. 5581--5586.

\bibitem{Parsapour2021}
M.~Parsapour and D.~Kuli{\'c}, ``Recovery-matrix inverse optimal control for
  deterministic feedforward-feedback controllers,'' in \emph{2021 American
  Control Conference (ACC)}, 2021, pp. 4765--4770.

\bibitem{Yu2021}
C.~Yu, Y.~Li, H.~Fang, and J.~Chen, ``System identification approach for
  inverse optimal control of finite-horizon linear quadratic regulators,''
  \emph{Automatica}, vol. 129, p. 109636, 2021.

\bibitem{Panchea2017}
A.~M. Panchea and N.~Ramdani, ``{Inverse Parametric Optimization in a
  Set-Membership Error-in-Variables Framework},'' \emph{IEEE Transactions on
  Automatic Control}, vol.~62, no.~12, pp. 6536--6543, Dec 2017.

\bibitem{Zhang2018}
H.~Zhang, J.~Umenberger, and X.~Hu, ``Inverse optimal control for discrete-time
  finite-horizon {Linear Quadratic Regulators},'' \emph{Automatica}, vol. 110,
  p. 108593, 2019.

\bibitem{Zhang2019}
H.~Zhang, Y.~Li, and X.~Hu, ``{Inverse Optimal Control for Finite-Horizon
  Discrete-time Linear Quadratic Regulator Under Noisy Output},'' in \emph{2019
  IEEE 58th Conference on Decision and Control (CDC)}, 2019, pp. 6663--6668.

\bibitem{Levine2012}
S.~Levine and V.~Koltun, ``Continuous inverse optimal control with locally
  optimal examples,'' in \emph{Proceedings of the 29th International Conference
  on Machine Learning (ICML-12)}, 2012, pp. 41--48.

\bibitem{jin_2021_pontryagin}
W.~Jin, Z.~Wang, Z.~Yang, and S.~Mou, ``{Pontryagin Differentiable Programming:
  An End-to-End Learning and Control Framework},'' in \emph{Advances in Neural
  Information Processing Systems}, H.~Larochelle, M.~Ranzato, R.~Hadsell,
  M.~Balcan, and H.~Lin, Eds., vol.~33.\hskip 1em plus 0.5em minus 0.4em\relax
  Curran Associates, Inc., 2020, pp. 7979--7992.

\bibitem{Ng2000}
A.~Y. Ng, S.~J. Russell \emph{et~al.}, ``Algorithms for inverse reinforcement
  learning.'' in \emph{ICML}, 2000, pp. 663--670.

\bibitem{Wulfmeier2015}
M.~Wulfmeier, P.~Ondruska, and I.~Posner, ``Deep inverse reinforcement
  learning,'' \emph{arXiv preprint arXiv:1507.04888}, 2015.

\bibitem{Ziebart2008}
B.~D. Ziebart, A.~L. Maas, J.~A. Bagnell, and A.~K. Dey, ``Maximum entropy
  inverse reinforcement learning,'' in \emph{AAAI Conference on Artificial
  Intelligence}, 2008, pp. 1433--1438.

\bibitem{Mombaur2010}
K.~Mombaur, A.~Truong, and J.-P. Laumond, ``From human to humanoid
  locomotion---an inverse optimal control approach,'' \emph{Autonomous robots},
  vol.~28, no.~3, pp. 369--383, 2010.

\bibitem{Aghasadeghi2014}
N.~Aghasadeghi and T.~Bretl, ``Inverse optimal control for differentially flat
  systems with application to locomotion modeling,'' in \emph{Robotics and
  Automation (ICRA), 2014 IEEE International Conference on}, May 2014, pp.
  6018--6025.

\bibitem{Puydupin2012}
A.-S. Puydupin-Jamin, M.~Johnson, and T.~Bretl, ``A convex approach to inverse
  optimal control and its application to modeling human locomotion,'' in
  \emph{Robotics and Automation (ICRA), 2012 IEEE International Conference on},
  May 2012, pp. 531--536.

\bibitem{Jin2018}
W.~Jin, D.~Kuli{\'c}, S.~Mou, and S.~Hirche, ``Inverse optimal control from
  incomplete trajectory observations,'' \emph{The International Journal of
  Robotics Research}, vol.~40, no. 6-7, pp. 848--865, 2021.

\bibitem{Jin2019}
W.~Jin, D.~Kuli{\'c}, J.~F.-S. Lin, S.~Mou, and S.~Hirche, ``Inverse optimal
  control for multiphase cost functions,'' \emph{IEEE Transactions on
  Robotics}, vol.~35, no.~6, pp. 1387--1398, 2019.

\bibitem{Molloy2022}
T.~L. Molloy, J.~I. Charaja, S.~Hohmann, and T.~Perez, \emph{Inverse optimal
  control and inverse noncooperative dynamic game theory}.\hskip 1em plus 0.5em
  minus 0.4em\relax Springer, 2022.

\bibitem{Johnson2013}
M.~Johnson, N.~Aghasadeghi, and T.~Bretl, ``Inverse optimal control for
  deterministic continuous-time nonlinear systems,'' in \emph{Decision and
  Control (CDC), 2013 IEEE 52nd Annual Conference on}, Dec 2013, pp.
  2906--2913.

\bibitem{Molloy2020}
T.~L. Molloy, J.~J. Ford, and T.~Perez, ``Online inverse optimal control for
  control-constrained discrete-time systems on finite and infinite horizons,''
  \emph{Automatica}, vol. 120, p. 109109, 2020.

\bibitem{Molloy2016}
T.~L. Molloy, D.~Tsai, J.~J. Ford, and T.~Perez, ``Discrete-time inverse
  optimal control with partial-state information: A soft-optimality approach
  with constrained state estimation,'' in \emph{Decision and Control (CDC),
  2016 IEEE 55th Annual Conference on}, Las Vegas, NV, Dec 2016.

\bibitem{Molloy2018b}
T.~L. Molloy, J.~J. Ford, and T.~Perez, ``{Online Inverse Optimal Control on
  Infinite Horizons},'' in \emph{2018 IEEE Conference on Decision and Control
  (CDC)}, 2018, pp. 1663--1668.

\bibitem{Self2019}
R.~Self, M.~Harlan, and R.~Kamalapurkar, ``Online inverse reinforcement
  learning for nonlinear systems,'' in \emph{2019 IEEE Conference on Control
  Technology and Applications (CCTA)}, 2019, pp. 296--301.

\bibitem{Self2020}
R.~Self, S.~M.~N. Mahmud, K.~Hareland, and R.~Kamalapurkar, ``Online inverse
  reinforcement learning with limited data,'' in \emph{2020 59th IEEE
  Conference on Decision and Control (CDC)}, 2020, pp. 603--608.

\bibitem{Self2020a}
R.~Self, M.~Abudia, and R.~Kamalapurkar, ``Online inverse reinforcement
  learning for systems with disturbances,'' in \emph{2020 American Control
  Conference (ACC)}, 2020, pp. 1118--1123.

\bibitem{Self2021}
R.~Self, K.~Coleman, H.~Bai, and R.~Kamalapurkar, ``Online observer-based
  inverse reinforcement learning,'' \emph{IEEE Control Systems Letters},
  vol.~5, no.~6, pp. 1922--1927, 2021.

\bibitem{Self2022}
R.~Self, M.~Abudia, S.~N. Mahmud, and R.~Kamalapurkar, ``Model-based inverse
  reinforcement learning for deterministic systems,'' \emph{Automatica}, vol.
  140, p. 110242, 2022.

\bibitem{cleach_2020_lucidgames}
S.~Le~Cleac'h, M.~Schwager, and Z.~Manchester, ``{LUCIDGames: Online Unscented
  Inverse Dynamic Games for Adaptive Trajectory Prediction and Planning},''
  \emph{IEEE Robotics and Automation Letters}, vol.~6, no.~3, pp. 5485--5492,
  2021.

\bibitem{sarkka2023bayesian}
S.~S{\"a}rkk{\"a} and L.~Svensson, \emph{{Bayesian Filtering and Smoothing}},
  2nd~ed.\hskip 1em plus 0.5em minus 0.4em\relax Cambridge University Press,
  2023.

\bibitem{reif1999stochastic}
K.~Reif, S.~Gunther, E.~Yaz, and R.~Unbehauen, ``{Stochastic stability of the
  discrete-time extended Kalman filter},'' \emph{IEEE Transactions on Automatic
  control}, vol.~44, no.~4, pp. 714--728, 1999.

\bibitem{Molloy2018}
T.~L. Molloy, J.~J. Ford, and T.~Perez, ``Finite-horizon inverse optimal
  control for discrete-time nonlinear systems,'' \emph{Automatica}, vol.~87,
  pp. 442 -- 446, 2018.

\end{thebibliography}

\end{document}